\documentclass[preprint,1p,number,12pt,times]{elsarticle}
\usepackage{amsmath,amsthm,amssymb}
\usepackage{tikz}
\usepackage{graphicx}
\usepackage{hyperref}

\newcommand{\bx}{\mathbf{x}}
\newcommand{\by}{\mathbf{y}}
\newcommand{\bn}{\mathbf{n}}
\newcommand{\cO}{\mathcal{O}}
\newcommand{\real}{\mathbb{R}}
\newcommand{\complex}{\mathbb{C}}
\newcommand{\znat}{\mathbb{Z}}

\newcommand{\eff}{\mathrm{eff}}
\newcommand{\ext}{\mathrm{(ext)}}

\newcommand{\abs}[1]{\left|{#1}\right|}
\newcommand{\norm}[1]{\left\|{#1}\right\|}
\newcommand{\M}[1]{\left({#1}\right)}

\newcommand{\Mcb}[1]{\left\{{#1}\right\}}
\newcommand{\ceil}[1]{\lceil{#1}\rceil}
\newcommand{\conj}[1]{\overline{#1}}

\newcommand{\rb}[1]{\raisebox{1em}{\rotatebox{90}{#1}}}
\newcommand{\rbb}[1]{\raisebox{3em}{\rotatebox{90}{#1}}}

\newtheorem{remark}{Remark}
\newtheorem{theorem}{Theorem}
\newcommand{\secref}[1]{\S\ref{#1}}
\newcommand{\figref}[1]{Figure~\ref{#1}}
\newcommand{\thmref}[1]{Theorem~\ref{#1}}

\newlength{\cbheight}
\newcommand{\colorbar}{\raisebox{0.5em}{\includegraphics[height=\cbheight]{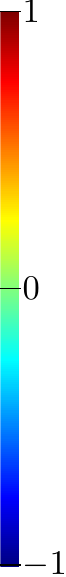}}}

\journal{Wave Motion}
\begin{document}

\begin{frontmatter}

\title{Exterior cloaking with active sources in\\ two dimensional acoustics}

\author[ut]{Fernando Guevara Vasquez\corref{cor}}
\ead{fguevara@math.utah.edu}

\author[ut]{Graeme W. Milton}
\ead{milton@math.utah.edu}

\author[ut]{Daniel Onofrei}
\ead{onofrei@math.utah.edu}

\address[ut]{Mathematics Department\\ 
University of Utah\\ 
155 S 1400 E RM 233\\
Salt Lake City UT 84112, USA.}

\cortext[cor]{Corresponding author.}

\begin{abstract}
We cloak a region from a known incident wave by surrounding the region with
three or more devices that cancel out the field in the cloaked region without
significantly radiating waves. Since very little waves reach scatterers 
within the cloaked region, the scattered field is small and the
scatterers are for all practical purposes undetectable. The devices are
multipolar point sources that can be determined from Green's formula and
an addition theorem for Hankel functions. The cloaking devices are
exterior to the cloaked region.
\end{abstract}

\begin{keyword}
Active cloaking \sep Acoustic waves \sep Helmholtz equation \sep Green's
formula
\MSC 74J05 
\sep 35J05 
\end{keyword}

\end{frontmatter}

\section{Introduction}

Interest in cloaking has surged, as reflected in the many recent reviews
\cite{Chen:2010:ACT,Greenleaf:2009:CDE,Alu:2008:PMC}. We introduced a
new kind of cloaking for the two-dimensional Helmholtz equation
\cite{Vasquez:2009:AEC, Vasquez:2009:BEC}. Our approach uses {\em active
sources} (cloaking devices) to hide objects placed in an external
region. The advantages of our approach are: (a) by the superposition
principle a cloak can be designed for a broad band of frequencies
(excluding discretely many frequencies where the object being cloaked,
if non-absorbing, ``resonates'') and (b) the cloak does not need
materials with extreme properties which are hard to realize and
dispersive, as is the case in most transformation based cloaking
strategies (see e.g.  \cite{Greenleaf:2003:ACC, Pendry:2006:CEM,
Leonhardt:2006:OCM, Farhat:2009:CBW, Greenleaf:2009:CDE} -- though an
exception is \cite{Leonhardt:2009:BIN}). A significant drawback of our
approach is that we assume full knowledge of the incident field. Also
the active sources contain a monopole term which may be problematic for
applications in electromagnetism.

The problem of finding source distributions for cloaking is clearly
ill-posed in the sense that if it admits one solution then it admits
infinitely many solutions.  In \cite{Vasquez:2009:AEC, Vasquez:2009:BEC}
we computed particular solutions involving three point-like devices by
solving a constrained least-squares problem with the singular value
decomposition (SVD).  In an effort to explain rigorously our previous
results, we use the Green representation theorem for the Helmholtz
equation (in short {\em Green's formula}, see e.g
\cite{Colton:1983:IEM}), to derive explicitly a particular solution in
terms of the incident field (\thmref{thm:green}).  Another cloaking
method based on Green's formula is the active interior cloak introduced
by Miller \cite{Miller:2007:PC}, which uses single and double layer
potentials to cancel the incident field inside a closed curve. The same
principle is used in active sound control \cite{Peterson:2007:ACS,
Ffowcs:1984:RLA} to suppress noise.  Here we apply an addition theorem
for Hankel functions to replace the source distribution on a curve by a
few active sources, effectively connecting the cloaked region with the
exterior. The sources we get are, as the ansatz used in
\cite{Vasquez:2009:AEC, Vasquez:2009:BEC}, multipolar point-like
sources: they are given as a series of cylindrical radiating solutions
to the Helmholtz equation, centered at a few points. We also establish
convergence of the series to the fields required for cloaking and give a
specific configuration of sources (\secref{sec:example}) similar to the
one we found empirically in \cite{Vasquez:2009:AEC, Vasquez:2009:BEC}.

Other methods for obtaining exterior cloaks include those based on
complementary media  \cite{Lai:2009:CMI}, surface plasmonic resonances
\cite{Silvereinha:2008:CMA}, anomalous resonances in the vicinity of a
superlens \cite{Milton:2006:CEA,Nicorovici:2007:OCT,Milton:2008:SFG} and
waveguides \cite{Smolyaninov:2009:AME}.

A different idea is that of illusion optics \cite{Lai:2009:IOO} where
the goal is to hide an object and make it appear as another object.  The
Green's formula based approach that we present here also explains why
this can be done using active devices, as was recently observed
numerically \cite{Zheng:2010:EOC}. We give a way of explicitly
constructing the devices to such effect, without the need for solving a
least-squares problem (see Remark~\ref{rem:illusion}).

We work in the frequency domain at a fixed angular frequency
$\omega$. In a medium with constant speed of propagation
$c$, the wave field $u(\bx,\omega)$ satisfies the Helmholtz equation 
\begin{equation}
 \Delta u + k^2 u = 0, ~~ \text{for $\bx \in \real^2$},
 \label{eq:helm}
\end{equation}
where $k = 2\pi / \lambda$ is the wavenumber and $\lambda = 2\pi c /
\omega$ is the wavelength. For simplicity we drop the dependency on the
frequency and write $u(\bx) \equiv u(\bx,\omega)$.

\begin{remark}
We assume that the frequency $\omega$ is not a resonant frequency of the
scatterer we want to hide. Resonant frequencies are left for future studies.
\end{remark}

\begin{remark}
 Calling a cloak ``active'' can be ambiguous as one can refer to a cloak
 that can hide active sources \cite{Greenleaf:2007:FWI} or a cloak that
 uses active sources to hide objects
 \cite{Miller:2007:PC,Vasquez:2009:AEC,Vasquez:2009:BEC}. The cloaking
 method we present here is ``active'' in both senses as we use active
 devices to hide objects and in some simple situations it is possible to hide 
 sources (see Remark~\ref{rem:ext}).  However here we focus only on
 cloaking scatterers.  
\end{remark}

\section{Green's formula cloaks} 
\label{sec:construction} 
We present the active interior cloak \cite{Miller:2007:PC} and using an
addition theorem for Hankel functions show how this cloak can be
replaced by a few multipolar sources (\secref{sec:active}). The price to
pay for using a finite number of sources is that the cloaked region and
the region from which the object is invisible are smaller compared to
the active interior cloak. We then make some geometric considerations
(\secref{sec:example}) to show that with the particular approach of
\secref{sec:active} we need three or more sources to get a non-empty
cloaked region. 

\subsection{Active interior cloak}
\label{sec:miller}
Denote by $D$ the region of $\real^2$ that we wish to cloak from a known
incident field $u_i$. We also assume from now on that $D$ is a simply
connected bounded region of class $C^2$. The arguments in
\secref{sec:construction} can be easily generalized to the case where
$D$ is composed of several simply connected components. In order to
cloak $D$, we construct a solution $u_d$ to the Helmholtz equation
\eqref{eq:helm} (in $\real^2$ excluding the boundary of $D$) such that 
\begin{equation}
 u_d(\bx) = 
 \begin{cases}
 -u_i(\bx) & \text{for $\bx \in D$}\\
 0 & \text{otherwise.}
 \end{cases}
 \label{eq:ud}
\end{equation}
Hence the total field $u_d + u_i$ vanishes in $D$ and is
indistinguishable from $u_i$ outside $D$. If a scatterer is placed
inside $D$ the scattered field $u_s$ resulting from the incident field
$u_d + u_i$ is zero.  Assuming $u_i$ is an analytic solution to the
Helmholtz equation \eqref{eq:helm} inside $D$, a field $u_d$ satisfying
\eqref{eq:ud} can be constructed using Green's formula (see e.g.
\cite{Colton:1983:IEM})
\begin{equation}
u_d(\bx) = \int_{\partial D} dS_\by \Mcb { - (\bn(\by) \cdot \nabla_\by
u_i(\by))
G(\bx,\by) + u_i(\by) \bn(\by) \cdot \nabla_\by G(\bx,\by) },
\label{eq:green}
\end{equation}
where $\bn(\by)$ is the unit outward normal to $D$ at a point $\by$ on
the boundary $\partial D$ and the Green's function for the two
dimensional wave equation is
\begin{equation} 
 G(\bx,\by) = \frac{i}{4} H_0^{(1)}(k\abs{\bx-\by}),
\end{equation}
where $H^{(1)}_n$ is the $n-$th Hankel function of the first kind
\cite[\S 9]{Abramowitz:1972:HMF}. The first term in the integrand in
\eqref{eq:green}  can be interpreted as the potential due to a
distribution of monopoles on the boundary $\partial D$ (the single
layer) while the second term can be interpreted as the potential due to
a distribution of dipoles oriented normal to $\partial D$ (the double
layer).

In the frequency domain, the cloaking scheme we obtain is the same
proposed by Miller \cite{Miller:2007:PC}, where the single and double
layer potentials in \eqref{eq:green} are simulated with many sources
completely surrounding the cloaked region $D$. We give an example of the
field $u_d$ generated by Green's formula \eqref{eq:green} in
\figref{fig:green_ud}. The effect of this active interior cloak can be seen in
\figref{fig:green}, where it is used to hide a kite shaped scatterer
\cite{Colton:1998:IAE} with homogeneous Dirichlet boundary conditions.
The field is virtually zero inside the cloaked region, so there is very
little scattered waves to detect the object.

\begin{figure}
\begin{center}
 \includegraphics[width=0.4\textwidth]{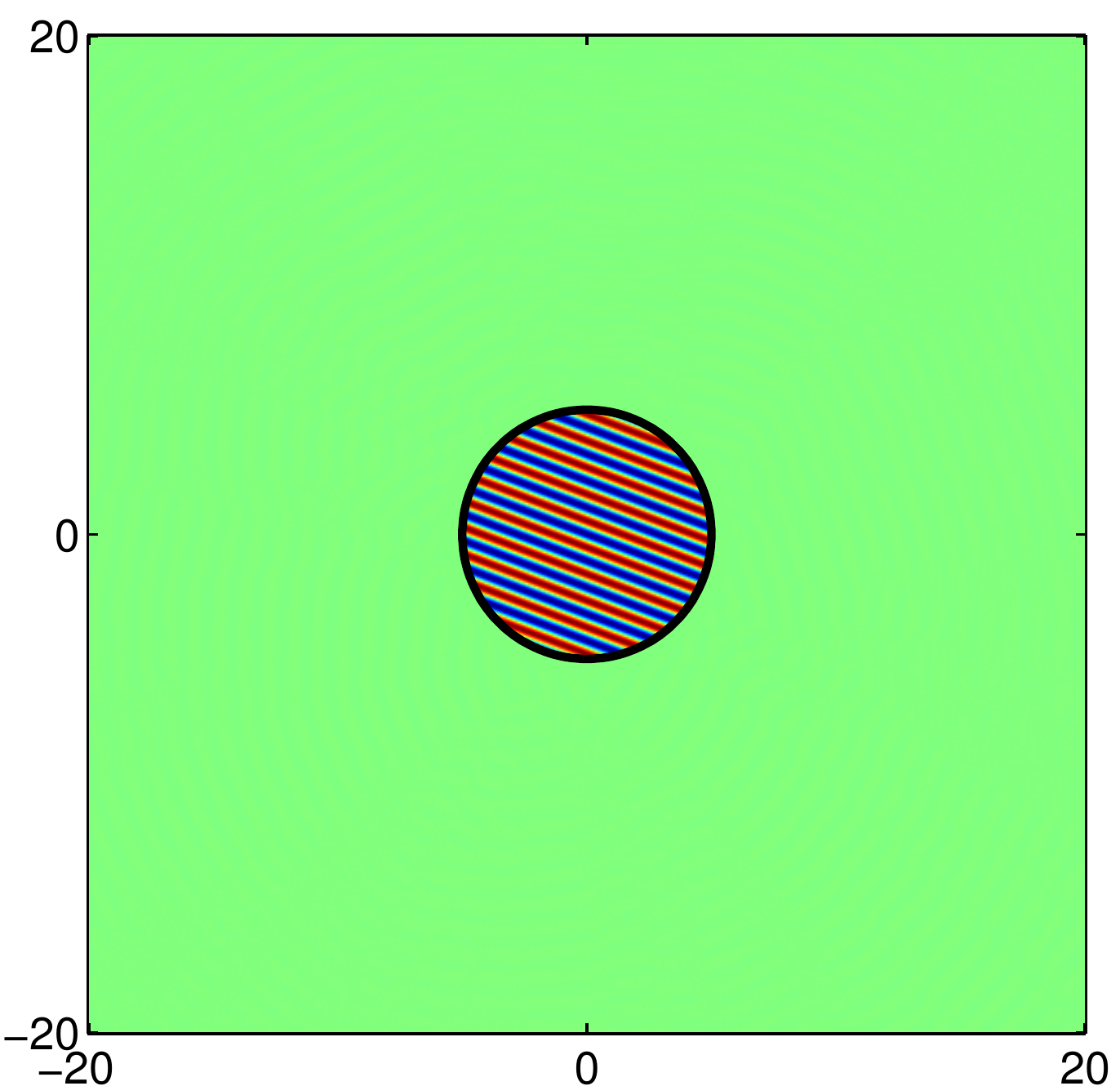}
 \colorbar
\end{center}
\caption{The field $u_d$ generated by Green's formula \eqref{eq:green}
The region $D$ is the disk in thick lines with radius $5\lambda$ and
the incident field is a plane wave with angle $5\pi/13$. The field $u_d$
is very close to the incident field inside $D$ and zero outside $D$. The
integral in \eqref{eq:green} is evaluated with the trapezoidal rule on
$2^8$ equally spaced points on $\partial D$. The axis units are in
wavelengths $\lambda$.} \label{fig:green_ud}
\end{figure}

\begin{figure}
\begin{tabular}{ccc}
 \includegraphics[width=0.4\textwidth]{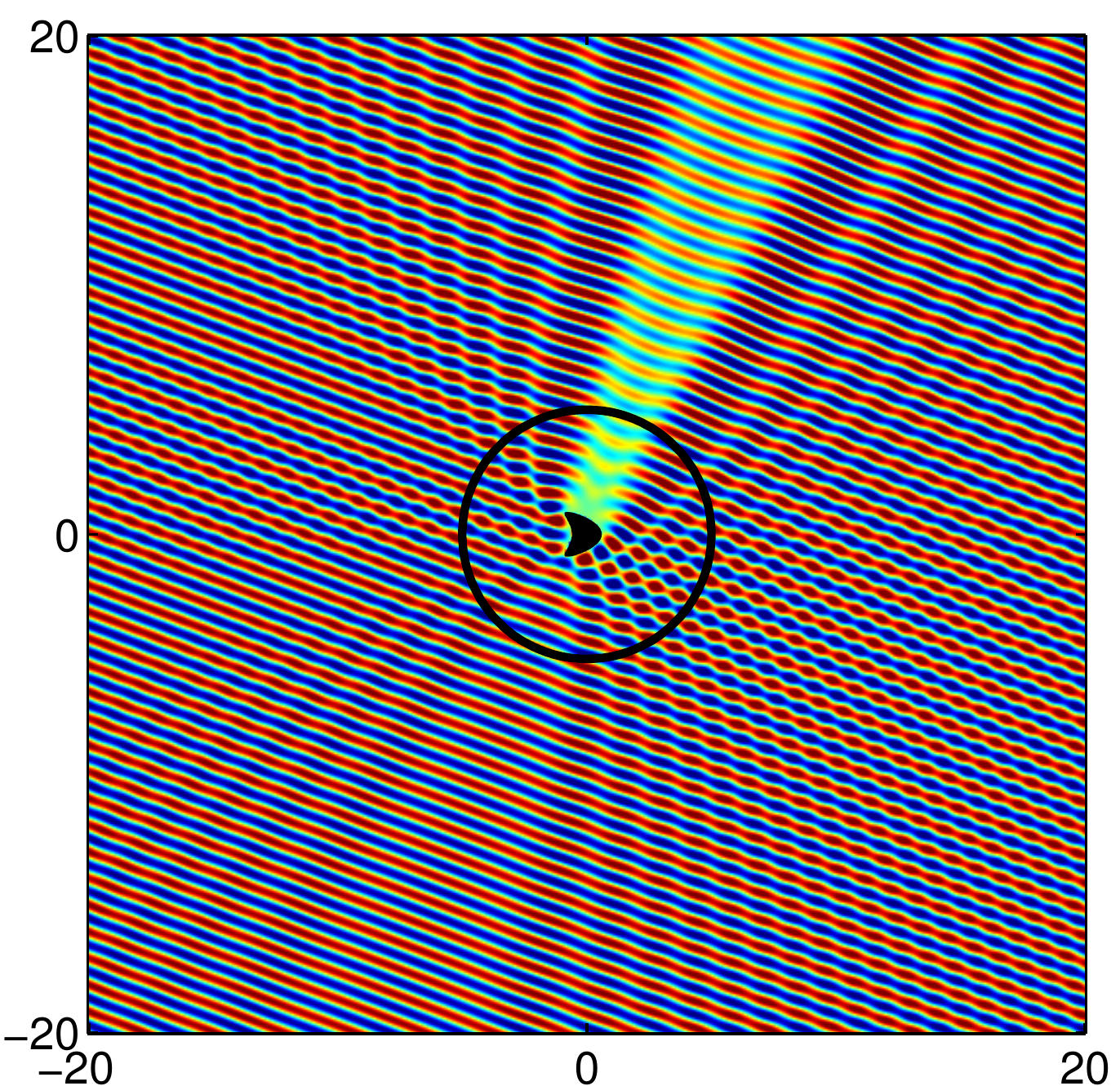} &
 \includegraphics[width=0.4\textwidth]{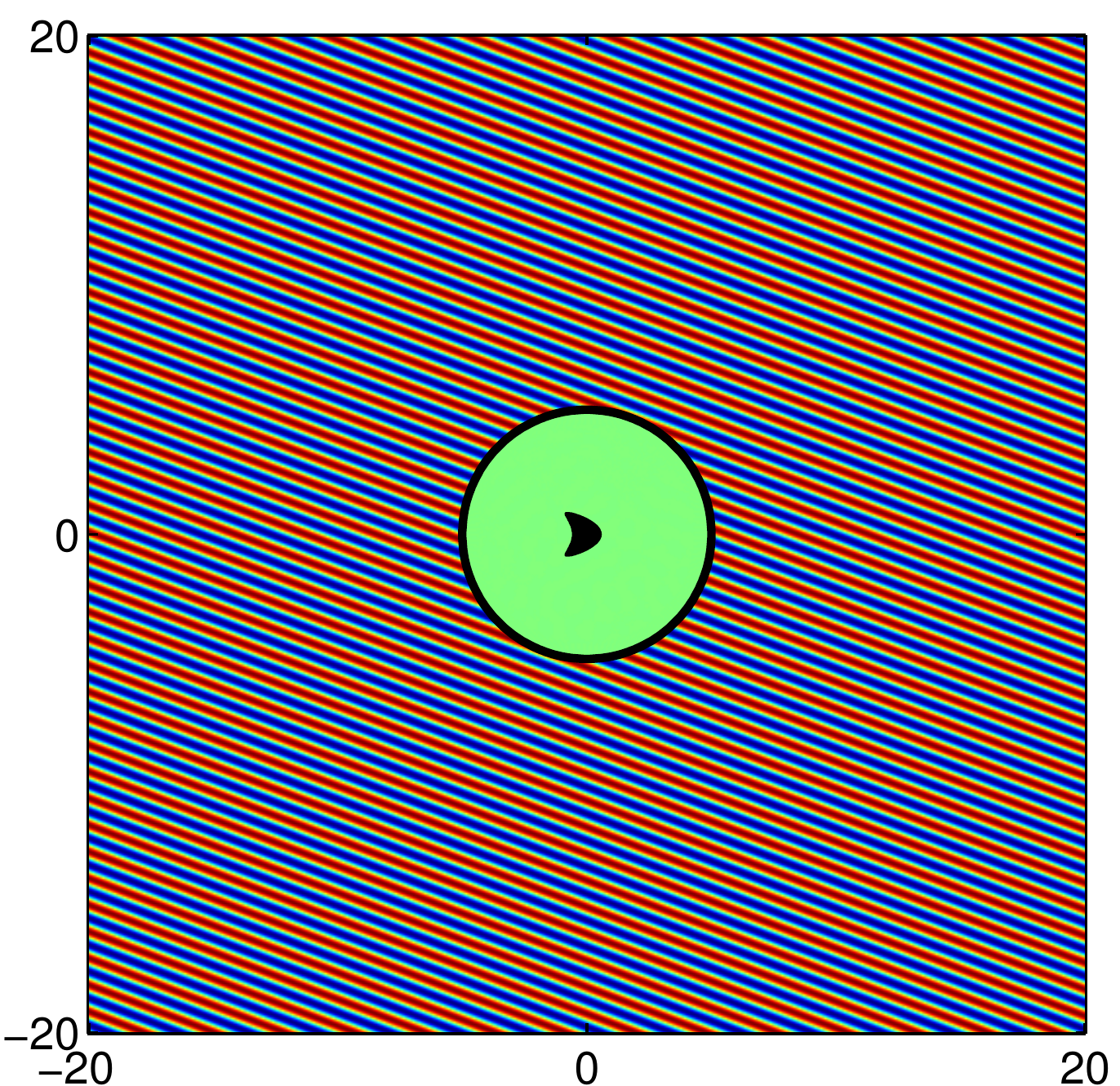} &
 \colorbar\\
 (a) & (b) &
\end{tabular}
\caption{Active interior cloak (a) inactive and (b) active. The region $D$
where Green's formula is applied is the circle of radius $5\lambda$ in
thick lines. The field inside the cloaked region $D$ is virtually zero,
so that outside the cloaked region the field is indistinguishable from
the incident plane wave with direction $5\pi/13$. The axis units are in
wavelengths $\lambda$.}
\label{fig:green}
\end{figure}

\subsection{Active exterior cloak}
\label{sec:active}
Our aim is to replace the single and double layer potentials on
$\partial D$ appearing in \eqref{eq:green} by a few $n_{dev}$ multipolar
sources (what we  call ``cloaking devices'') located at some points
$\bx_j\notin \partial D$. The advantage being that the cloaked region is
no longer completely enclosed by a surface. The field generated by such
sources can be written formally as
\begin{equation}
 u_d^{\ext}(\bx) = \sum_{j=1}^{n_{dev}} \sum_{m=-\infty}^\infty b_{j,m}
 V_m(\bx-\bx_j),
 \label{eq:uext}
\end{equation}
where the coefficients $b_{j,m} \in \complex$ are to be determined and
\[
 V_m(\bx) \equiv H_m^{(1)} (k \abs{\bx} ) \exp [ i m \arg(\bx) ]
\]
are radiating cylindrical waves. Here $\arg(\bx)$ denotes the
counterclockwise oriented angle from the vector $(1,0)$ to the vector
$\bx$. By a radiating solution to the Helmholtz equation, we mean it
satisfies the Sommerfeld radiation condition (see e.g.
\cite{Colton:1998:IAE}).

In \cite{Vasquez:2009:AEC,Vasquez:2009:BEC} we presented a numerical
scheme based on the singular value decomposition (SVD) to compute the
coefficients $b_{j,m}$ in a way that $u_d^{\ext}$ approximates  $u_d$ as
defined in \eqref{eq:ud}. We give next in \thmref{thm:green} one way of
obtaining these coefficients explicitly such that $u_d^{\ext}(\bx) =
u_d(\bx)$ for $\bx$ in a certain region $R \subset \real^2$, effectively
reducing the cloaked region to $D\cap R$.

Let us assign to each source $\bx_j$ a segment $\partial D_j$ of the
boundary. These segments are chosen such that they partition $\partial
D$ and  $\partial D_i \cap \partial D_j$ is empty or a single point when
$i\neq j$.  An example of this setup is given in \figref{fig:setup}. The
coefficients $b_{j,m}$ are given next.

\begin{theorem}
\label{thm:green}
Multipolar sources located at $\bx_j \notin \partial D$,
$j=1,\ldots,n_{dev}$, can be used to reproduce the active interior cloak in the
region
\[
 R = \bigcap_{j=1}^{n_{dev}} \Mcb{ \bx \in \real^2 ~~\Big|~~ |\bx-\bx_j|
 > \sup_{\by \in \partial D_j} \abs{\by - \bx_j} }.
\]
The coefficients $b_{j,m}$ in
\eqref{eq:uext} such that $u_d^{\ext}(\bx) = u_d(\bx)$ for $\bx \in R$
are given by
\begin{equation}
\begin{aligned}
 b_{j,m} = \int_{\partial D_j} dS_\by &\{ 
 \M{ -\bn(\by) \cdot \nabla_\by u_i(\by) } \conj{U_m(\by-\bx_j)}\\
 &+ u_i(\by) \bn(\by) \cdot \nabla_\by \conj{U_m(\by-\bx_j)} \}
\end{aligned}
\label{eq:bjm}
\end{equation}
for $j=1,\ldots,n_{ext}$ and for $m\in \znat$. Here $U_m(\bx)$ are
entire cylindrical waves,
\[
 U_m(\bx) \equiv J_m(k\abs{\bx}) \exp[ i m \arg(\bx) ].
\]
\end{theorem}
\begin{proof}
Clearly, the integral over the whole boundary $\partial D$ in Green's
formula \eqref{eq:green} can be written as the sum of the integrals
over the segments $\partial D_j$. For segment $\partial D_j$, we use
Graf's addition formula \cite[\S 9.1.79]{Abramowitz:1972:HMF} to express
the Green's function $G(\bx,\by)$ as a superposition of multipolar
sources located at $\bx_j$,
\begin{equation}
 \begin{aligned}
 G(\bx,\by) & = \frac{i}{4} H_0^{(1)}( k \abs{\bx - \bx_j - (\by -
 \bx_j)} )\\
 & = \frac{i}{4} \sum_{m=-\infty}^\infty V_m(\bx-\bx_j)
 \conj{U_m(\by-\bx_j)},
 \end{aligned}
 \label{eq:graf}
\end{equation}
where the series converges absolutely and uniformly on compact subsets
of $\abs{\bx - \bx_j} > \abs{\by - \bx_j}$ (this can be seen by e.g.
adapting Theorem 2.10 in \cite{Colton:1998:IAE} to two dimensions).
Splitting the integral in Green's formula \eqref{eq:uext} into
integrals over $\partial D_j$ and using the expansion \eqref{eq:graf} we
get for $\bx \in R$,
\begin{equation}
 \begin{aligned}
 u_d^{\ext}(\bx) = \sum_{j=1}^{n_{dev}} \int_{\partial D_j} dS_\by &\{ 
 \M{ -\bn(\by) \cdot \nabla_\by u_i(\by) } \sum_{m=-\infty}^\infty
 V_m(\bx-\bx_j)\conj{U_m(\by-\bx_j)}\\
 &+ u_i(\by) \bn(\by) \cdot \nabla_\by \sum_{m=-\infty}^\infty
 V_m(\bx-\bx_j)\conj{U_m(\by-\bx_j)} \}.
 \end{aligned}
 \label{eq:expansion}
\end{equation}

The desired result \eqref{eq:bjm} can be obtained by
rearranging the infinite sum and the integral. For the first term in the
integrals in \eqref{eq:expansion}, the uniform convergence of the series
\eqref{eq:graf} for $\by \in \partial D_j$ allows us to switch the
infinite sum and the integral. The second term involves the gradient
\begin{equation}
\begin{aligned}
 \nabla_\by \conj{U_m(\by - \bx_j)} &= k \frac{\by
 -\bx_j}{\abs{\by-\bx_j}} J_m'(k\abs{\by-\bx_j})
 \exp[-im\arg(\by-\bx_j)]\\
 &+ J_m(k\abs{\by-\bx_j}) (-im \exp[-im\arg(\by-\bx_j)] ) \nabla_\by
 \arg(\by - \bx_j).
\end{aligned}
\label{eq:gradubar}
\end{equation}
Since we assumed $\bx_j \notin \partial D$, the gradient
\begin{equation}
\nabla_\by \arg(\by - \bx_j) = 
\frac{(\by-\bx_j)^\perp}{\abs{\by-\bx_j}^2},
~~ \text{with}~ \bx^{\perp} \equiv \begin{bmatrix} -x_2 \\ x_1 \end{bmatrix},
\end{equation}
is bounded for $\by \in \partial D_j$. Using the series representation for
Bessel functions (see e.g. \cite[\S 9.3.1 ]{Abramowitz:1972:HMF} and
\cite[\S 3.4]{Colton:1998:IAE}) we can get the estimates
\begin{equation}
 \begin{aligned}
  J_n(t) &= \frac{t^n}{2^n n!} (1+\cO(1/n)), & J_n'(t) &=
  \frac{t^{n-1}}{2^n (n-1)!}(1+\cO(1/n))\\
  H_n^{(1)}(t) &=\frac{2^n(n-1)!}{\pi i t^n}(1+\cO(1/n))
 \end{aligned}
\label{eq:besselest}
\end{equation}
valid for $t>0$ as $n\to\infty$, uniformly on compact subsets of $(0,\infty)$.
Using \eqref{eq:besselest} and the expression for the gradient
\eqref{eq:gradubar} we can estimate the terms
\begin{equation}
 V_m(\bx-\bx_j) \nabla_\by \conj{U_m(\by - \bx_j)} =
 \cO\M{\frac{\abs{\by -\bx_j}^{m-1}}{\abs{\bx-\bx_j}^m}} 
+\cO\M{m\frac{\abs{\by-\bx_j}^m}{\abs{\bx-\bx_j}^m}}
\end{equation}
as $m\to\infty$ uniformly on compact subsets of $\abs{\bx - \bx_j} >
\abs{\by-\bx_j}$. Therefore the series in the second integrand of
\eqref{eq:expansion} converges absolutely and uniformly in $\partial
D_j$ and the infinite sum and the integral can be switched.
\end{proof}

Note that \thmref{thm:green} does not guarantee that the effective
cloaked region $D\cap R$ is not empty. However we do have that the
device's field vanishes far away from the devices (i.e.
$u_d^{\ext}(\bx) = 0$ for $|\bx|$ large enough) because $\real^2
\backslash R$ is bounded. Later in \secref{sec:example} we give a
specific configuration where $D\cap R \neq \emptyset$.

\begin{remark}
\label{rem:ext}
 In order to guarantee that the field $u_d(\bx)$ in Green's formula
 \eqref{eq:green} vanishes outside $D$, we need an analytic incident
 field $u_i(\bx)$ inside $D$. If the field $u_i(\bx)$ is a radiating and
 $C^2$ solution to the Helmholtz equation \eqref{eq:helm} outside $D$
 (as is the case when there are sources and non-resonant scatterers
 inside $D$), then Green's formula \eqref{eq:green} converges outside
 $D$ to $u_i(\bx)$ and inside $D$ to zero (see e.g.
 \cite{Colton:1983:IEM}).  This is the principle behind noise
 suppression \cite{Ffowcs:1984:RLA} and could be used to cloak active
 sources and scatterers in $D$, assuming they are known. 
 \end{remark}
 \begin{remark}
 \label{rem:illusion}
 Clearly the same Green's formula approach can
 be used to create illusions with active sources \cite{Zheng:2010:EOC}:
 one can simultaneously cloak an object and generate waves that
 correspond to the scattering from a completely different object (the
 ``virtual object''). All we need is knowledge of the scattered
 field $u_s^{virt}(\bx)$ generated by hitting the virtual object with
 the incident $u_i(\bx)$. Since $u_s^{virt}(\bx)$ is a radiating
 solution to the Helmholtz equation,  to achieve illusion
 simply subtract $u_s^{virt}$ from $u_i$ in \eqref{eq:bjm} (this is
 assuming that $u_s^{virt}(\bx)$ is $C^2$ outside $D$).
\end{remark}

\begin{remark}
The proof of \thmref{thm:green} generalizes easily to the Helmholtz
equation in three dimensions.  We leave this generalization for future
studies.
\end{remark}

\begin{figure}
\begin{center}
\includegraphics[width=0.45\textwidth]{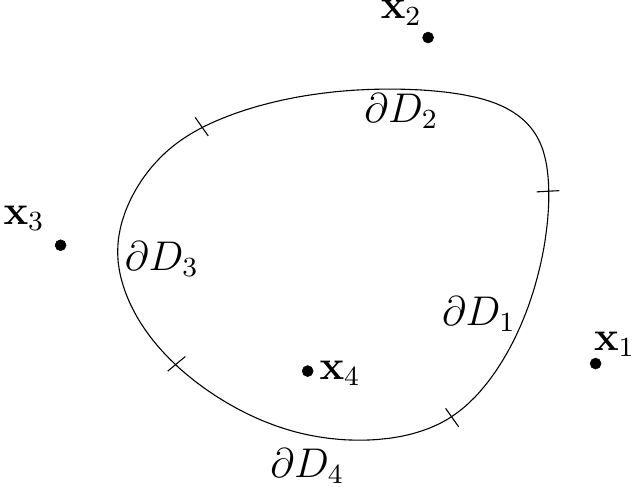}
\end{center}
\caption{Active exterior cloak construction. The contribution of portion
$\partial D_j$ to the single and double layer potentials in Green's
formula \eqref{eq:green} is replaced by a multipolar source located at
$\bx_j\notin \partial D$.}
\label{fig:setup}
\end{figure}

\subsection{An explicit example of an active exterior cloak}
\label{sec:example}
In this Section we describe one possible realization of the cloak
configuration presented in \thmref{thm:green}.  First notice that with
this particular method we need to have $n_{dev} \geq 3$ for a non-empty
effective cloaked region $D \cap R$ (with $D$ and $R$ given as in
\thmref{thm:green}). This is consistent with our numerical results in
\cite{Vasquez:2009:AEC,Vasquez:2009:BEC}, as we observed that at least
three devices are apparently needed to cloak plane waves with an
arbitrary direction of propagation.

To see that at least three devices are needed, first notice that
$\real^2 \backslash R = \bigcup_{j=1}^{n_{dev}} B_j$, where the $B_j$
are disks centered at the $j-$th device location $\bx_j$ and
$\partial D_j \subset B_j$.  We thus get
\[
 \partial D \subset \bigcup_{j=1}^{n_{dev}} B_j.
\]
If we have only one device ($n_{dev}=1$), then $D \subset \real^2
\backslash R$ and the effective cloaked region is empty. If we have two
devices and $D$ is simply connected, then we must also have $D \subset
\real^2 \backslash R$, as the union of two non-disjoint disks is simply
connected.

With three devices we take as an example the configuration shown in
\figref{fig:specific}. Here the devices are located at a distance
$\delta$ from the origin and are the vertices of an equilateral
triangle. The region $D$ where we apply Green's formula is the disk of
radius $\sigma$ centered at the origin. The circle $\partial D$ is
partitioned into three arcs $\partial D_j$, $j=1,2,3$ of identical
length which are chosen so that the distances $\sup_{\by \in \partial
D_j} \abs{\by - \bx_j}$ are equal for $j=1,2,3$. 

Simple geometric arguments show that the region $R$ of
\thmref{thm:green} is the complement of the union of the three disks in
gray in \figref{fig:specific}, with radius 
\begin{equation}
 r(\sigma,\delta) = ( (\sigma -\delta/2)^2 + 3\delta^2/4 )^{1/2},
 \label{eq:rsd}
\end{equation}
and centered at $\bx_j$, $j=1,2,3$.
To get an idea of the dimensions of the effective cloaked region $R\cap
D$ (in green in \figref{fig:specific}), we look at the radius of the
largest disk that can be inscribed inside. This disk has radius 
\begin{equation}
r_{\eff}(\sigma,\delta) = \delta - r(\sigma,\delta).
\label{eq:reff}
\end{equation}
Thus for fixed $\delta$ the largest
effective cloaked region is obtained when $\sigma = \delta /2$, which
corresponds to the case where the intersection of two of the disks in
gray in \figref{fig:specific} is a single point. Then the radius of the
largest disk that can be inscribed inside $R\cap D$ is
\begin{equation}
 r_{\eff}^*(\delta) = (1-\sqrt{3}/2) \delta \approx 0.13 \delta.
\label{eq:reffopt}
\end{equation}

\begin{figure}
\begin{center}
 \includegraphics[width=0.45\textwidth]{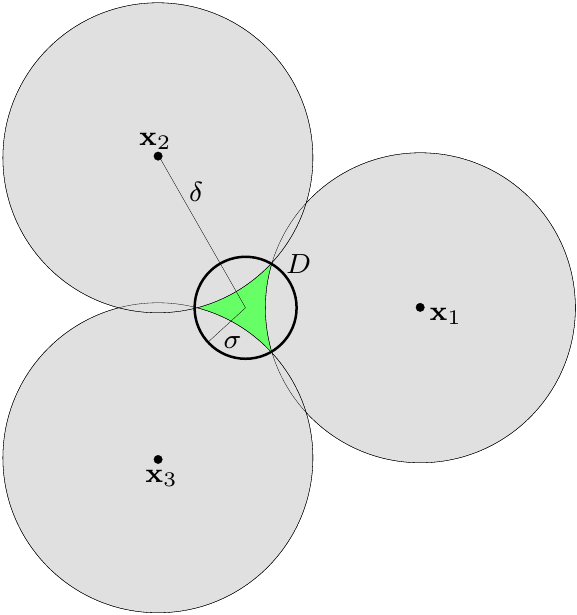}
\end{center}
\caption{A configuration for a cloak with three devices forming an
equilateral triangle such that $|\bx_j| = \delta$, $j=1,2,3$. The region
$D$ where we apply Green's formula is the disk of radius $\sigma$
centered at the origin. The region $\real^2 \backslash R$ appears in
gray. The green colored region is the effective cloaked region $D\cap
R$.}
\label{fig:specific}
\end{figure}

\section{Numerical experiments}
\label{sec:numexp}
We compare the Green's formula based method with the geometry described
in \secref{sec:example} and the maximal cloaked region size (i.e.
$\sigma = \delta/2$) to an SVD based method
\cite{Vasquez:2009:AEC,Vasquez:2009:BEC}. In the SVD based method three
distances are needed to describe the devices and the cloaked region: the
distance $\delta$ from the devices to the origin, the radius $\alpha$ of
the cloaked region and the radius $\gamma$ of the circle where we
enforce that the device's field be small. The numerical experiments in
\cite{Vasquez:2009:AEC,Vasquez:2009:BEC} were done with $\alpha =
\delta/5$ and $\gamma = 2 \delta$. To compare the cloaked regions for
both methods we choose
\[
 \alpha = r_{\eff}^*(\delta) = (1-\sqrt{3}/2) \delta
\]
and leave $\gamma$ unchanged. In this way the cloaked region for the SVD
method should be the largest disk that fits inside the effective cloaked
region $D \cap R$ of the Green's formula method.

We first compare in \figref{fig:comp} the device's fields and the total
field in the presence of a scatterer for both the SVD method and Green's
formula method. Here the devices are at a distance $\delta=10\lambda$
from the origin, with $k=1$ and $\lambda=2\pi$. In all our numerical
experiments the series \eqref{eq:uext} is truncated to $m=-M,\ldots,M$.
For the SVD method we follow the heuristic $M(\delta)=\ceil{(k
\delta/2)(1+\sqrt{3}/2)}$ in \cite{Vasquez:2009:BEC}, which for the
setup of \figref{fig:comp} gives $M(10\lambda)=59$. For comparison purposes
we used the same number of terms in the Green's formula method.  In both
methods, the device's field cancels out the incident field in region
near the origin without changing the incident plane wave. The region
where the total field  vanishes is larger for the Green's formula method
than for the SVD method. And for the former method, the cloaked region
seems larger than what is predicted by \secref{sec:example}.

The fields near the devices (the ``urchins'' in \figref{fig:comp}) are
very large as can be expected from the asymptotic behavior of Hankel
functions at the origin 
\[
 H^{(1)}_n(t) = \cO(t^{-|n|}) ~~ \text{as $t\to 0$ for 
 $n\in\znat \backslash \{0\}$.}
\]
The Green's formula allows us to replace each of the devices by a
closed curve containing the device and with appropriate single and
double layer potentials. The curves could be chosen as circles outside
of which the device's field has reasonable values (e.g. less than 100
times the magnitude of the incident field). For both methods these
circles do not touch, leaving ``throats'' connecting the cloaked region
to the exterior, so the cloaked region remains outside these
``extended'' devices. 

\begin{figure}
\begin{tabular}{c@{}ccc}
 & Green's formula & SVD &\\
 \rb{Device's field $u_d$} & 
 \includegraphics[width=0.4\textwidth]{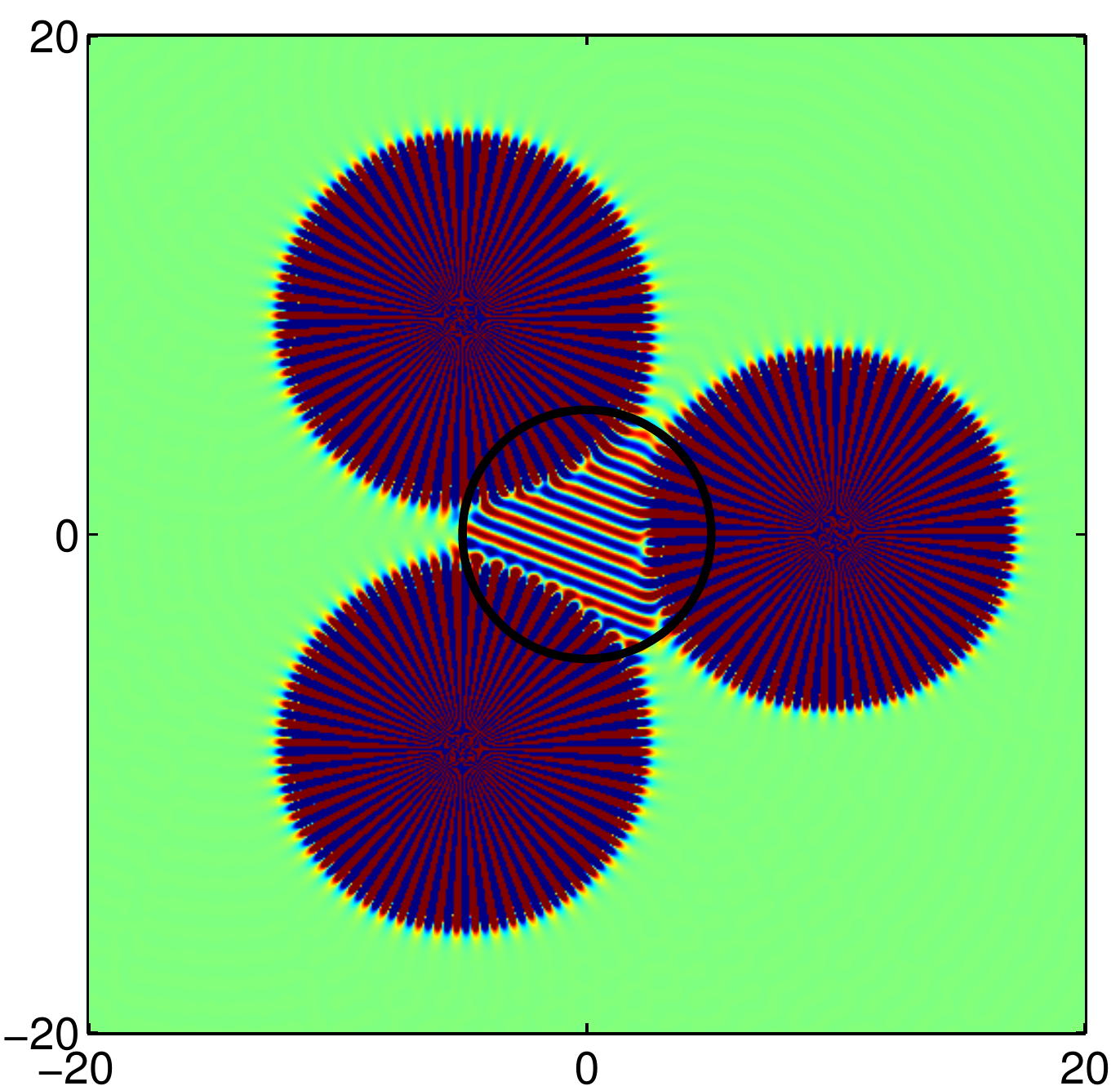} &
 \includegraphics[width=0.4\textwidth]{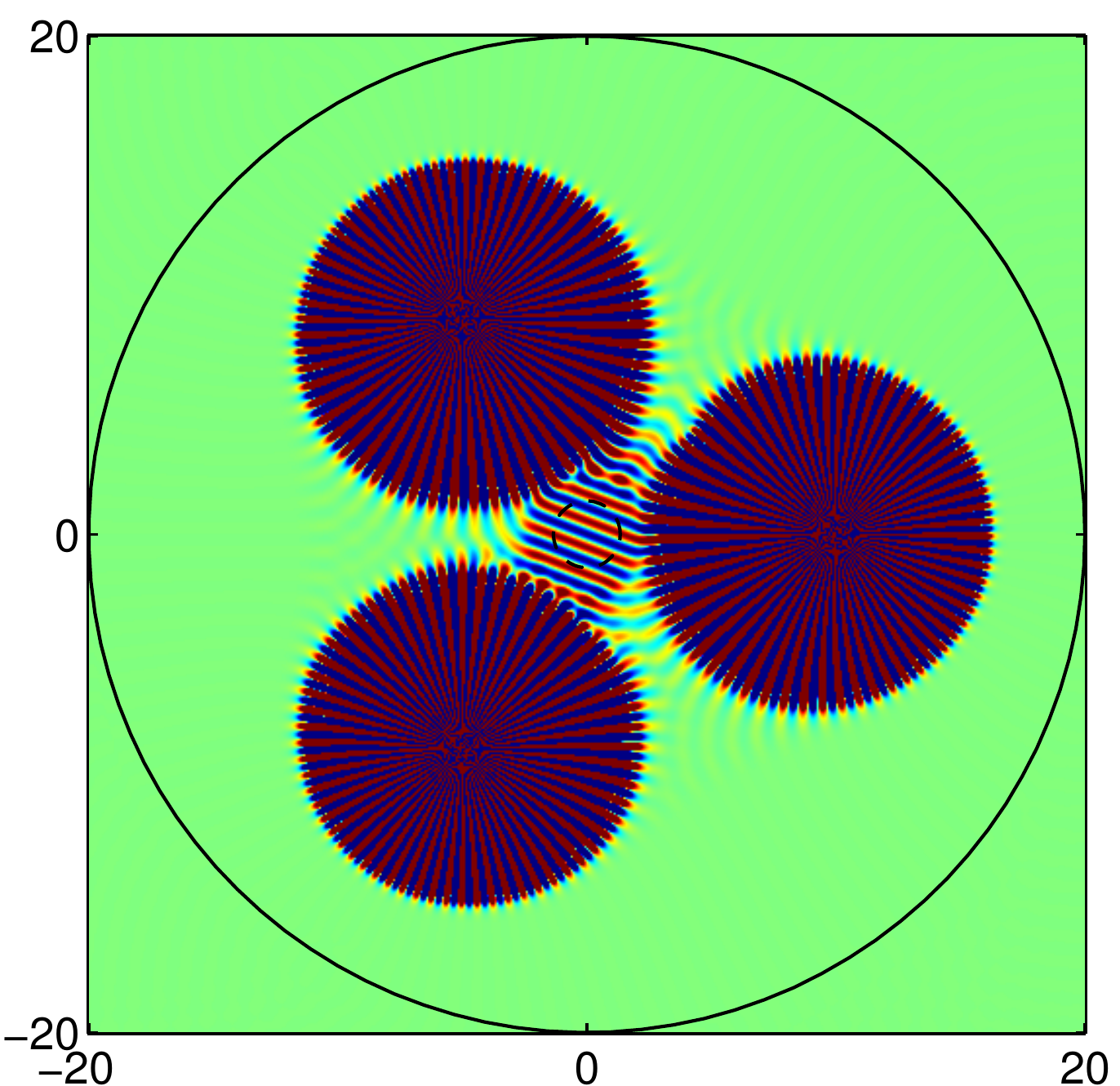} &
 \colorbar\\
 \rb{Total field $u_i + u_d + u_s$} &
 \includegraphics[width=0.4\textwidth]{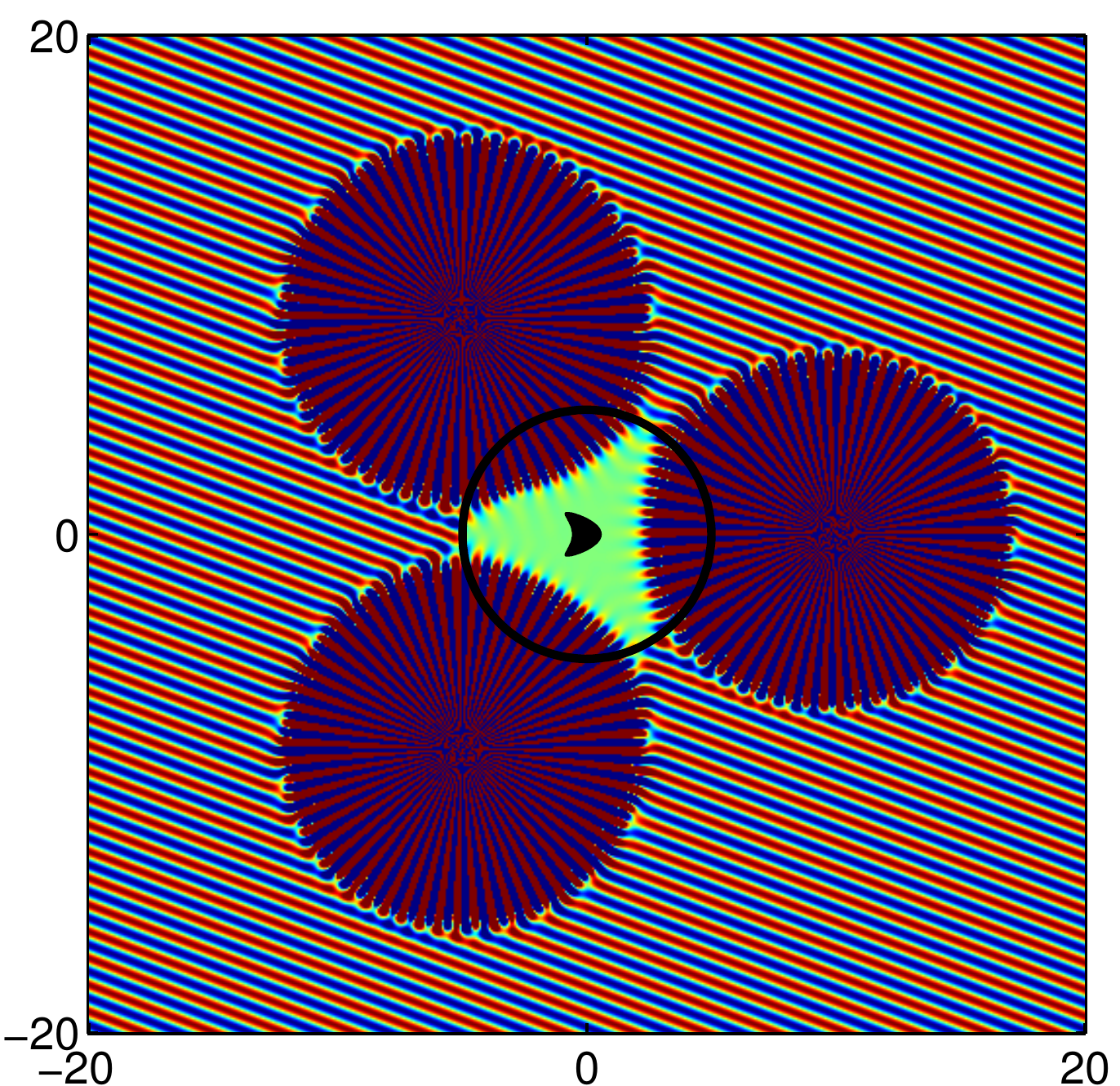} &
 \includegraphics[width=0.4\textwidth]{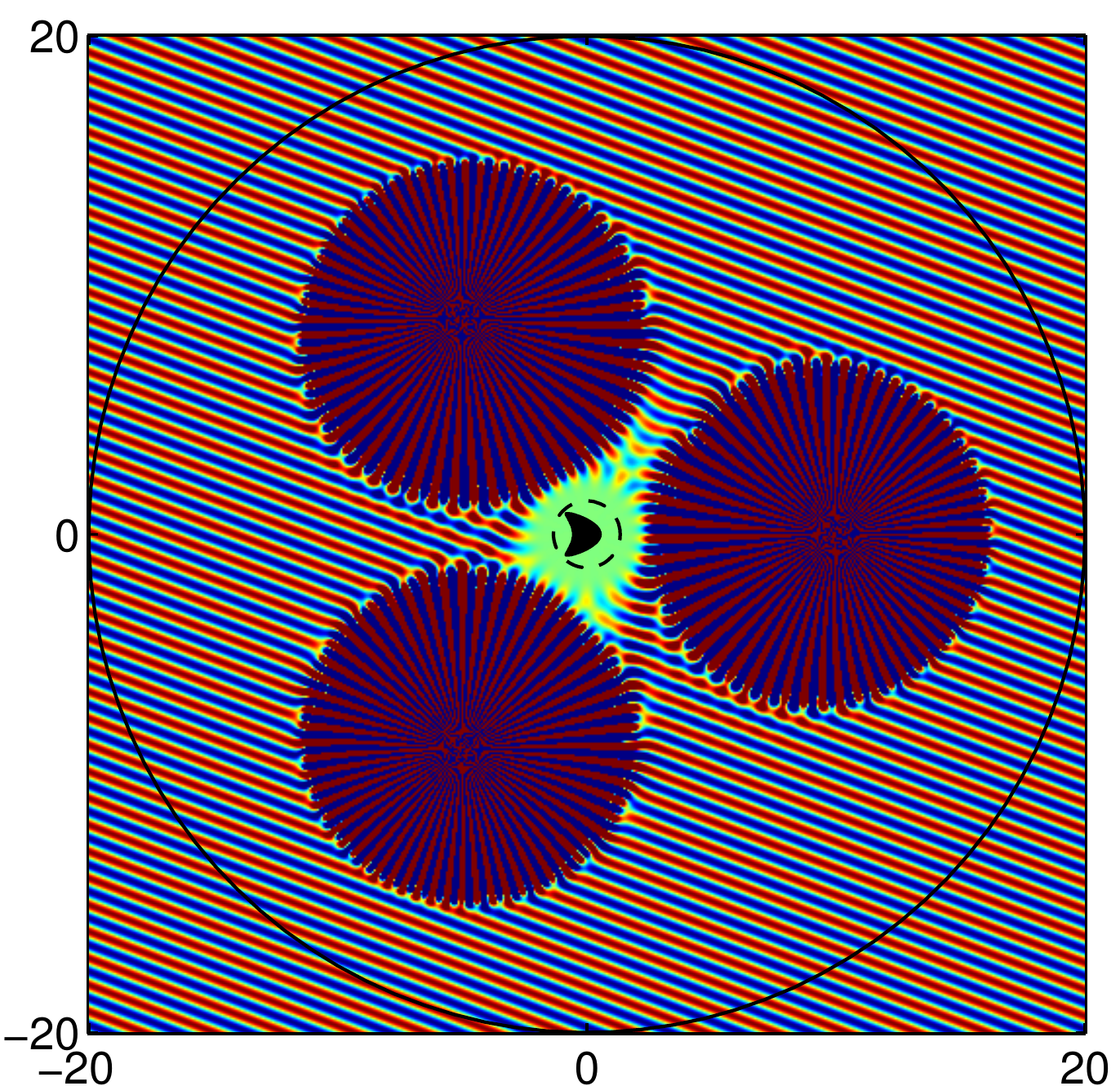} &
 \colorbar\\
\end{tabular}
\caption{Comparison of the Green's formula method (left column) and the
SVD based method \cite{Vasquez:2009:AEC,Vasquez:2009:BEC}
(right column). The circle in thick lines is where we apply Green's formula.
The dashed circle is where we enforce that the field be close to the
incident field and the solid circle is where we enforce that the
device's field is close to zero. The first row corresponds to the
device's field and the second row to the total field. } 
\label{fig:comp}
\end{figure}

We also considered larger cloaked regions in \figref{fig:perf},
keeping the same configuration as in \figref{fig:comp} but taking
$\delta \in [5,50]\lambda$. To evaluate the cloak performance we plot in
\figref{fig:perf} the quantity
\[
 \frac{\norm{u_i + u_d}_{L^2}}{\norm{u_i}_{L^2}},
 ~\text{on the circle $|\bx| = (1-\sqrt{3}/2)\delta$},
\]
which measures how well the device's field cancels the incident field
inside the cloaked region and
\[
 \frac{\norm{u_d}_{L^2}}{\norm{u_i}_{L^2}},
 ~\text{on the circle $|\bx| = 2\delta$},
\]
which measures how small the device's field is relative to the incident
field far away from the devices.\footnote{The values of $\|u_d\|/\|u_i\|$ 
reported in Figure 4b in \cite{Vasquez:2009:BEC} are, due to a
normalization mistake, up to a factor of $\sqrt{10}$ smaller than what
they should be. In the logarithmic scale
we use to display this quantity, the resulting shift is small
and our conclusions remain the same.} With the same number of terms
$M(\delta)$ the SVD method outperforms the Green's formula method in
both the cloaked region and far away from the devices. When $2M(\delta)$
terms are used for the Green's formula method, the relative errors
improve, but the convergence of the series in \eqref{eq:uext} seems
slower in the cloaked region than far away from the devices.

We estimate in \figref{fig:perfsz} the size of the ``extended'' devices
(i.e. the ``urchins'' where the device's field is large) by assuming
they are disks centered at the device's locations $\bx_i$. The disk
radius for the device centered at $\bx_i$ is estimated by finding the
closest intersection of the level set $|u_d(\bx)| = \beta$ with each of
the segments between $\bx_i$ and the other two devices. The quantity
plotted in \figref{fig:perfsz} is the maximum of these distances
rescaled by $\delta$ and is always below the radius at which the
extended cloaking devices do not leave gaps with the exterior (for both
cut-off values $\beta=5$ and $10$). The cost of having more terms in
Green's formula is that the extended devices leave narrower throats.
Presumably in the limit $M\to \infty$ these extended devices touch and
correspond to the region $\real^2 \backslash R$ in \thmref{thm:green}.

\begin{remark}
A natural question is whether changing the integrals over $\partial D_j$
in \eqref{eq:bjm} to integrals over subsets of $\partial D_j$ gives good
cloaking devices. Since in the region $R$ of \thmref{thm:green} the
device's field is identical to that of the active interior cloak, we can
reformulate the question as follows: does making small openings in $\partial D$
give a good active interior cloak? This is not the case because the
resulting active interior cloak has fields identical in $R$ to the
fields for the cloak taking all of
$\partial D$  into account minus the fields obtained by integrating
Green's formula \eqref{eq:green} on the portions of $\partial D$ that were
excluded. The excluded portions have a monopole and dipole distribution
that radiates and spoils the cloaking effect that we are after, even for
small openings.
\end{remark}

\begin{figure}
 \begin{center}
 \begin{tabular}{cc}
  $\|u_i + u_d\|/\|u_i\|$ on $|\bx| = (1-\sqrt{3}/2)\delta$ & $\|u_d\|/\|u_i\|$ on
  $|\bx| = 2\delta$\\
  \rbb{(percent)}\includegraphics[width=0.4\textwidth]{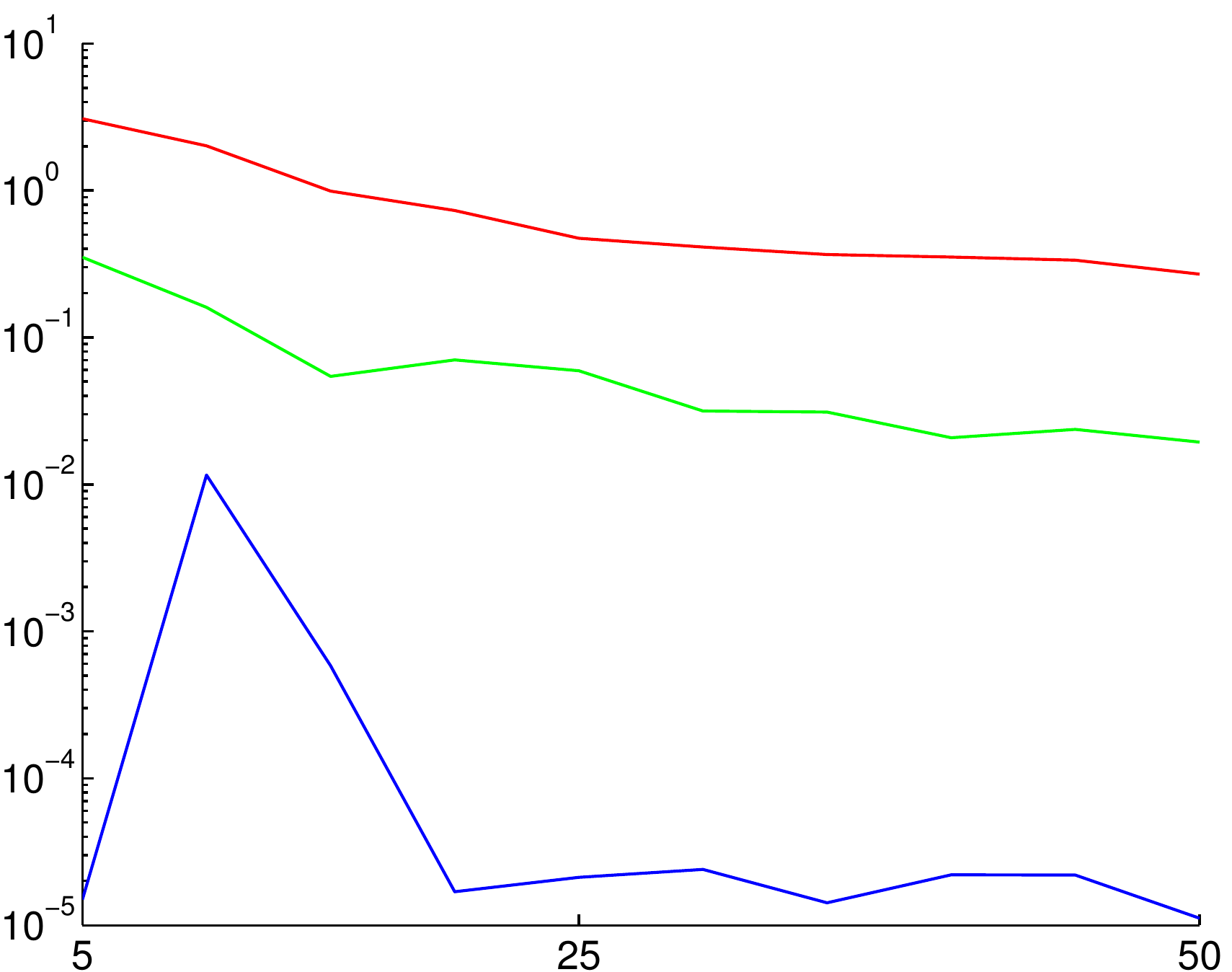} &
  \rbb{(percent)}\includegraphics[width=0.4\textwidth]{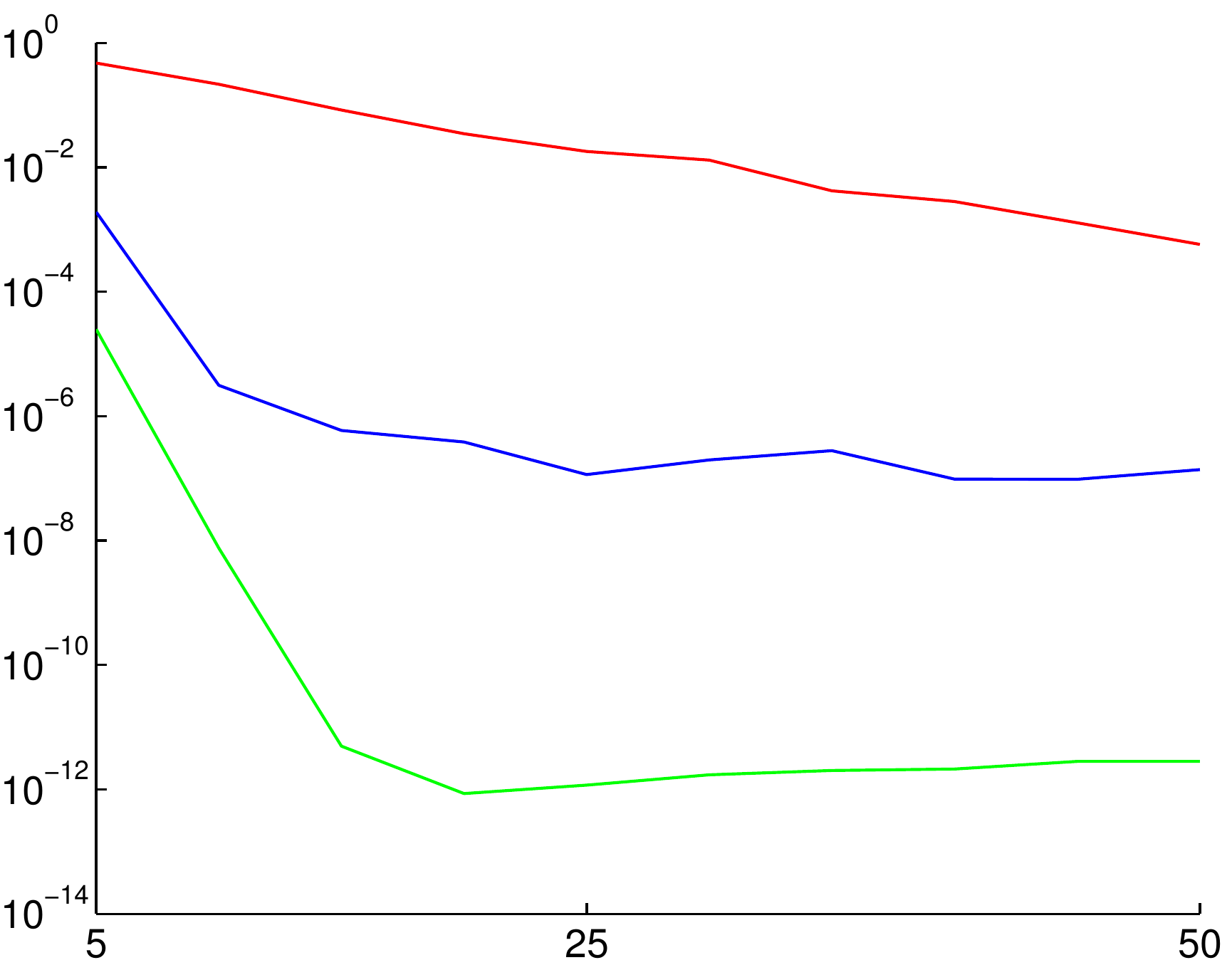}\\
  (a): $\delta$ (in $\lambda$) & (b): $\delta$ (in $\lambda$)
 \end{tabular}
 \end{center}
 \caption{Comparison of the cloak performance for the SVD method with
 $M(\delta)$ terms (blue) and the Green's identity method with
 $M(\delta)$ terms (red) and $2M(\delta)$ terms (green). In (a) we
 measure how small is the total field inside the cloaked region and in
 (b) how small is the device field far away from the devices.}
 \label{fig:perf}
\end{figure}

\begin{figure}
 \begin{center}
 \begin{tabular}{cc}
 cut-off $|u_d(\bx)|=100$ &  cut-off $|u_d(\bx)|=5$\\
 \rb{(device radius / $\delta$)}%
 \includegraphics[width=0.4\textwidth]{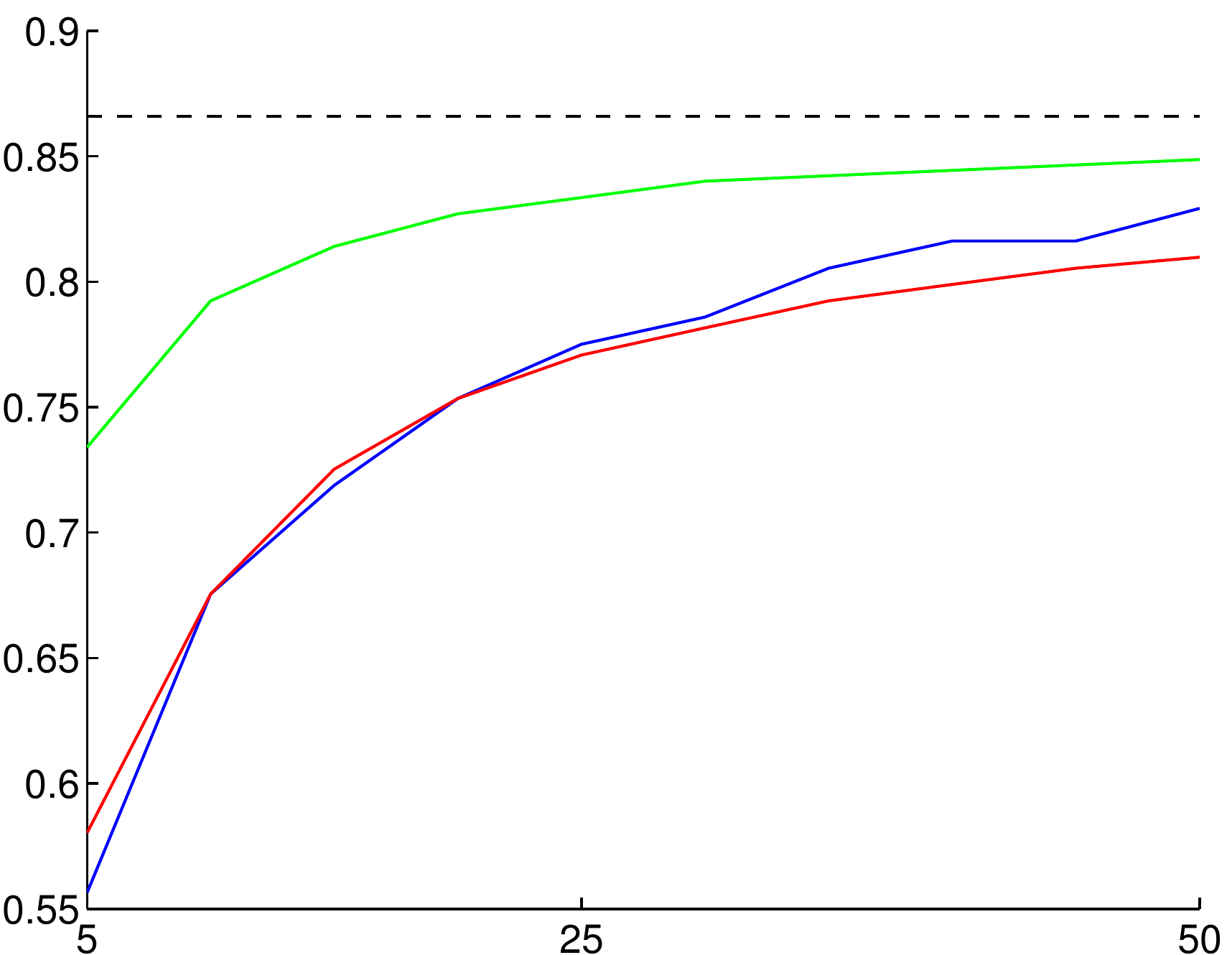} &
 \rb{(device radius / $\delta$)}%
 \includegraphics[width=0.4\textwidth]{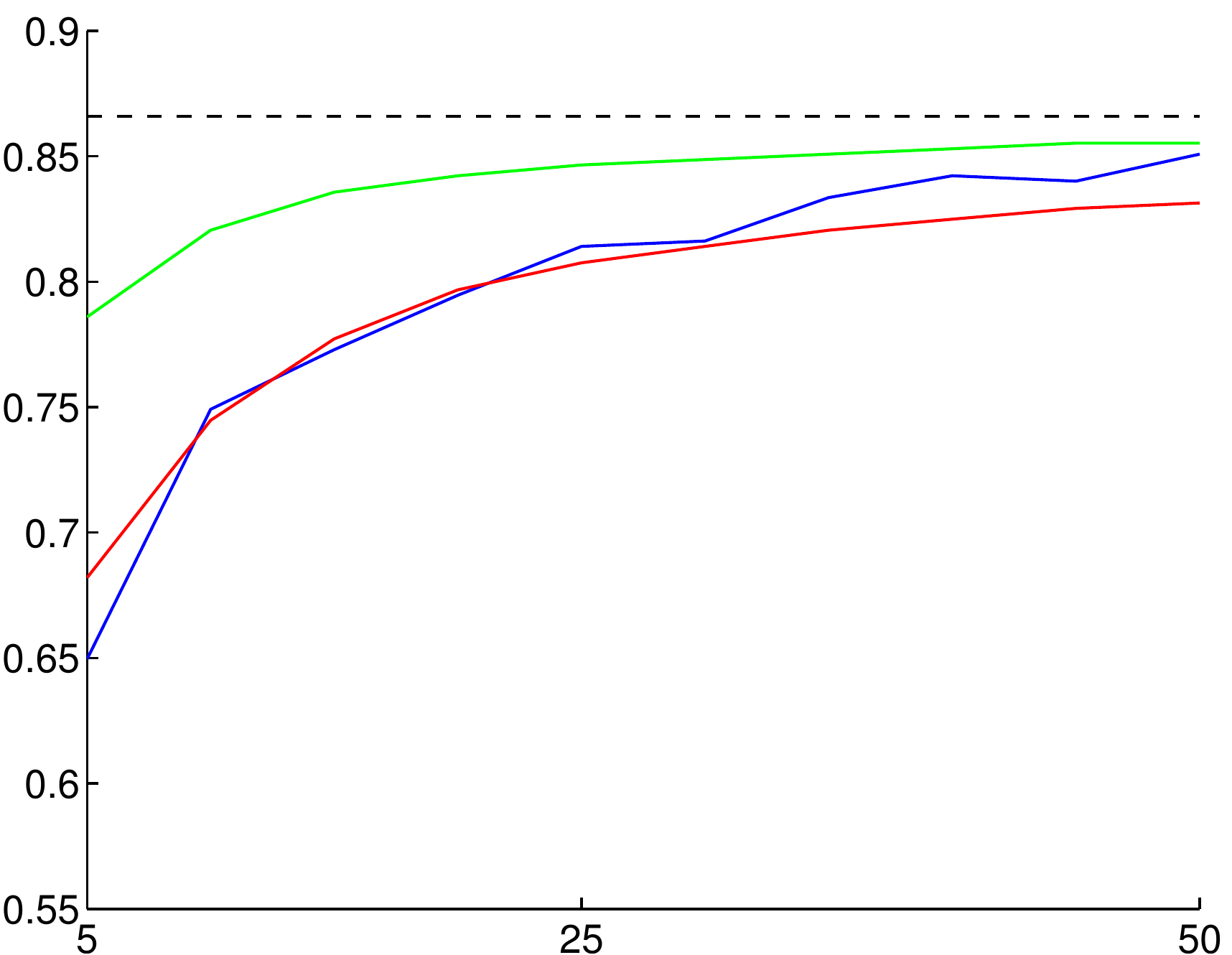}
 \\
 (a): $\delta$ (in $\lambda$) & (b): $\delta$ (in $\lambda$)
 \end{tabular}
 \end{center}
 \caption{Estimated device radius relative to $\delta$ for different
 values of $\delta$. The dotted line corresponds to the radius of the
 devices if they were touching circles. The SVD method with
 $M(\delta)$ terms is in blue, the Green's identity method with
 $M(\delta)$ terms is in red and with $2M(\delta)$ terms in green.}
 \label{fig:perfsz}
\end{figure}


\section*{Acknowledgments}
The authors are grateful for support from the National Science
Foundation through grant DMS-070978. GMW and FGV wish to thank the
Mathematical Sciences Research Institute, where this manuscript was
completed.

\bibliographystyle{abbrvnat}
\bibliography{wmbib}
\end{document}